\documentclass[pdflatex,sn-mathphys-num]{sn-jnl}% Math and Physical Sciences Numbered Reference Style
\usepackage{graphicx}%
\usepackage{multirow}%
\usepackage{amsmath,amssymb,amsfonts}%
\usepackage{amsthm}%
\usepackage{mathtools}%
\usepackage{mathrsfs}%
\usepackage[title]{appendix}%
\usepackage{xcolor}%
\usepackage{xpatch}%
\usepackage{textcomp}%
\usepackage{manyfoot}%
\usepackage{booktabs}%
\usepackage{algorithm}%
\usepackage{algorithmicx}%
\usepackage{algpseudocode}%
\usepackage{listings}%
\theoremstyle{thmstyleone}%
\newtheorem{theorem}{Theorem}%  meant for continuous numbers
\newtheorem{proposition}[theorem]{Proposition}%
\newtheorem{lemma}[theorem]{Lemma}% 

\theoremstyle{thmstyletwo}%

\theoremstyle{thmstylethree}%

\newcommand{\RR}{\mathbb{R}}

\newcommand{\ZZ}{\mathbb{Z}}
\newcommand{\NN}{\mathbb{N}}
\newcommand{\xdif}{\,\text{d}}
\DeclareMathOperator*{\Res}{Res}
\DeclarePairedDelimiter\floor{\lfloor}{\rfloor}

\makeatletter
\newcounter{proofpart}
\xpretocmd{\proof}{\setcounter{proofpart}{0}}{}{}
\newcommand{\proofpart}[1]{%
  \par
  \addvspace{\medskipamount}%
  \stepcounter{proofpart}%
  \noindent\emph{Part \theproofpart: #1}\par\nobreak\smallskip
  \@afterheading
}
\makeatother

\begin{document}

\title{Explicit Solutions in a Probabilistic Cellular Automaton Model for Self-Assembled Monolayer Degradation}


\author[1,2]{\fnm{Wes} \sur{Caldwell}}%\email{wcaldwe6@jh.edu}
\author[2]{\fnm{Ryan} \sur{Evans}}%\email{ryan.evans@nist.gov}
\author[1,2]{\fnm{Jacob} \sur{Majikes}}%\email{jacob.majikes@nist.gov}
\author*[2]{\fnm{Anthony} \sur{Kearsley}}\email{anthony.kearsley@nist.gov}
\affil[1]{Whiting School of Engineering, Johns Hopkins University, Baltimore, MD}
\affil*[2]{National Institute of Standards and Technology, Gaithersburg, MD}


\abstract{
Self-assembled monolayers (SAMs), formed by the adsorption of a monomolecular layer of adsorbate onto a surface, have emerged as leading candidates for platforms for electrochemical biosensors.
In order to characterize the working age of a SAM-based biosensor, it is necessary to understand the long-term degradation of the monolayer, with particular emphasis on the intermolecular forces present in the self-assembly process.
In this work, discrete- and continuous-time probabilistic cellular automaton (PCA) models for the degradation of SAMs involving adsorbate-adsorbate interactions are proposed in which adsorbate molecules desorb at a rate dependent only on the number of adsorbed neighbors.
Importantly, we make no assumptions on the relationship between degradation rate and adsorbed neighbors.
Using techniques from generating function theory and complex analysis, the expected degradation curve is solved for explicitly in one dimension and shown to be a sigmoid curve related to the Gompertz function.
}




\maketitle

\section{Introduction}

In the 21st century, self-assembled monolayers (SAMs) have been shown to be a fundamental tool in the construction and optimization of modern biosensors \cite{chakiSelfassembledMonolayersTunable2002, evansTransportPhenomenaBiological2020, evansDiffusionlimitedReactionsNanoscale2020, kimSelfassemblingBiomoleculesBiosensor2023}.
As a result, it is important to develop models to analyze the properties of SAMs over their functional lifetimes.
Monolayers are formed when molecules adsorb to the surface, and they degrade when the molecules desorb.
Traditional approaches to study the formation of monolayers use the Langmuir adsorption model \cite{langmuirConstitutionFundamentalProperties1917} to determine coverage, but this model assumes there are no local interactions between adsorbed molecules, so cannot say anything about the geometric structure of the formed monolayer.
Similarly, the most-used model for desorption dynamics, the Polanyi-Wigner equation, does not take into account adsorbate-adsorbate interactions \cite{ligterinkOverviewDesorptionParameters2023}.
Therefore, there is a need to develop models for monolayer dynamics that incorporates the geometry of the monolayer itself.

To this end, we propose a model for monolayer desorption based on probabilistic cellular automata (PCAs) that includes local interactions within the monolayer.
Cellular automata have a long history in the study of pattern formation in mathematical biology \cite{deutschCellularAutomatonModeling2017} dating back to the pioneering work of von Neumann in the 1960s \cite{vonNeumannTheorySelfreproducingAutomata1969}.
Cellular automata have the advantage of being very simple to define and very easy to simulate at the microscale while also being able to generate a wide range of morphogenetic phenomena at the macroscale.
We choose a PCA approach for several reasons: first, the underlying substrate, typically either gold or a semiconductor, to which the monolayer adsorbs, comes naturally equipped with a lattice structure of available sites, so a discrete lattice-based approach is natural.
Second, the requirement that each site acts independently with dynamics depending only on its neighbors is exactly the fundamental modeling assumption of PCA.
We may also interpret our PCA model as a Markov chain for each site with transition probabilities dependent on the state of each site's neighbors.

In one dimension, we may solve explicitly for the expected degraded area over time in terms of the parameters of the model (Theorems \ref{thm:main} and \ref{thm:general}).
On the other hand, in two dimensions, our methods cannot produce an explicit formula.
Our approach centers around the determination of the fundamental recurrence relation (\ref{eq:recurrence}), which in turn relies fundamentally on symmetries only present in the one-dimensional lattice: namely, that there is only one unique cluster of a given size.
One could in theory construct a similar recurrence relation over the explicit clusters in the two dimensional lattice, but this approach would be inherently constrained to the ability to enumerate clusters, which is a famously difficult open problem \cite{golombPolyominoesPuzzlesPatterns1994, grimmettPercolation1999, redelmeierCountingPolyominoesAnother1981}.
Therefore, for this work, we restrict our search for explicit solutions to the one dimensional case.

\begin{figure}
    \centering
    \includegraphics[width=0.75\linewidth]{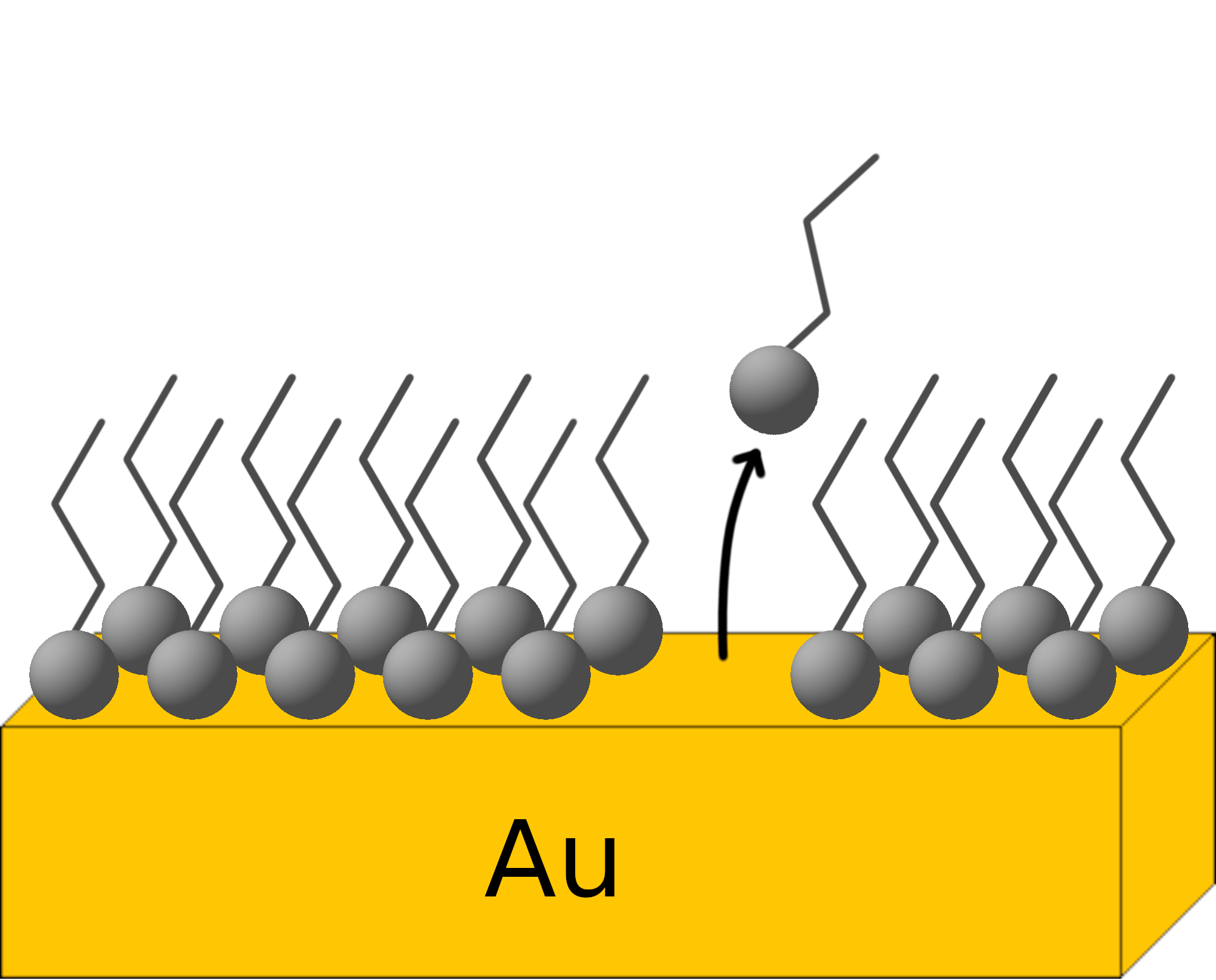}
    \caption{Desorption of SAM molecules from the substrate lattice.}
    \label{fig:desorption}
\end{figure}

\section{Degradation PCA}\label{sec:model}

\subsection{Model Definition}
We model the monolayer as the infinite one- or two-dimensional lattice $\ZZ^N$, for $N=1,2$, respectively, where each site either contains an adsorbate molecule, in which case we say the site is healthy, or does not contain an adsorbate molecule, in which case we say the site is degraded.
Let $s_i(t)$ denote the state of site $i \in \ZZ^N$ at time $t \in \NN$, where $s_i(t) = 0$ means that site $i$ is healthy (not degraded) at time $t$ and $s_i(t) = 1$ means that site $i$ is degraded at time $t$.
At time $t=0$, we suppose the monolayer is healthy, so that $s_i(0)=0$ for all $i \in \ZZ^N$.

Our model obeys the following two rules:
\begin{enumerate}
    \item We consider only possible desorption of molecules, so once a state becomes degraded it cannot become healthy again;
    \item When we step from time $t$ to time $t+1$, each healthy site degrades simultaneously with probability $p_n$, where $n$ is the number of degraded neighbors of the site.
\end{enumerate}
In the one-dimensional model, we must therefore specify the probabilities $p_0, p_1, p_2$, whereas in the two-dimensional model, we must specify the probabilities $p_0$ through $p_4$.
It is not required for the results in this paper, but in practice we would expect these probabilities to be increasing, since degradations are more likely the more degraded neighbors a site has \cite{quinonesIsothermModelsLocalized1996}.

From a biochemical standpoint, Rule 1 states that there is no readsorption from solution.
In practice this will never be exactly the case, but if the concentration of adsorbate remains sufficiently low throughout the degradation process, then Rule 1 is an appropriate idealization.
On the other hand, Rule 2 states that we may view the state of each site as a Markov chain with transition probability $0 \to 1$ depending on the state of the neighboring sites.
We may also view Rule 2 as a discretization of an underlying continuous process (see Section \ref{sec:cont}).

This model resembles a discrete version of the contact process from the study of interacting particle systems \cite{harrisContactInteractionsLattice1974, liggettInteractingParticleSystems2005}, as well as cellular automaton models for the diffusion of innovations \cite{boccaraModelingDiffusionInnovations1999, fuksExampleComputationDensity2018}.
From now on, we will consider only the one-dimensional model.

\begin{figure}
    \centering
    \includegraphics[width=0.4\linewidth]{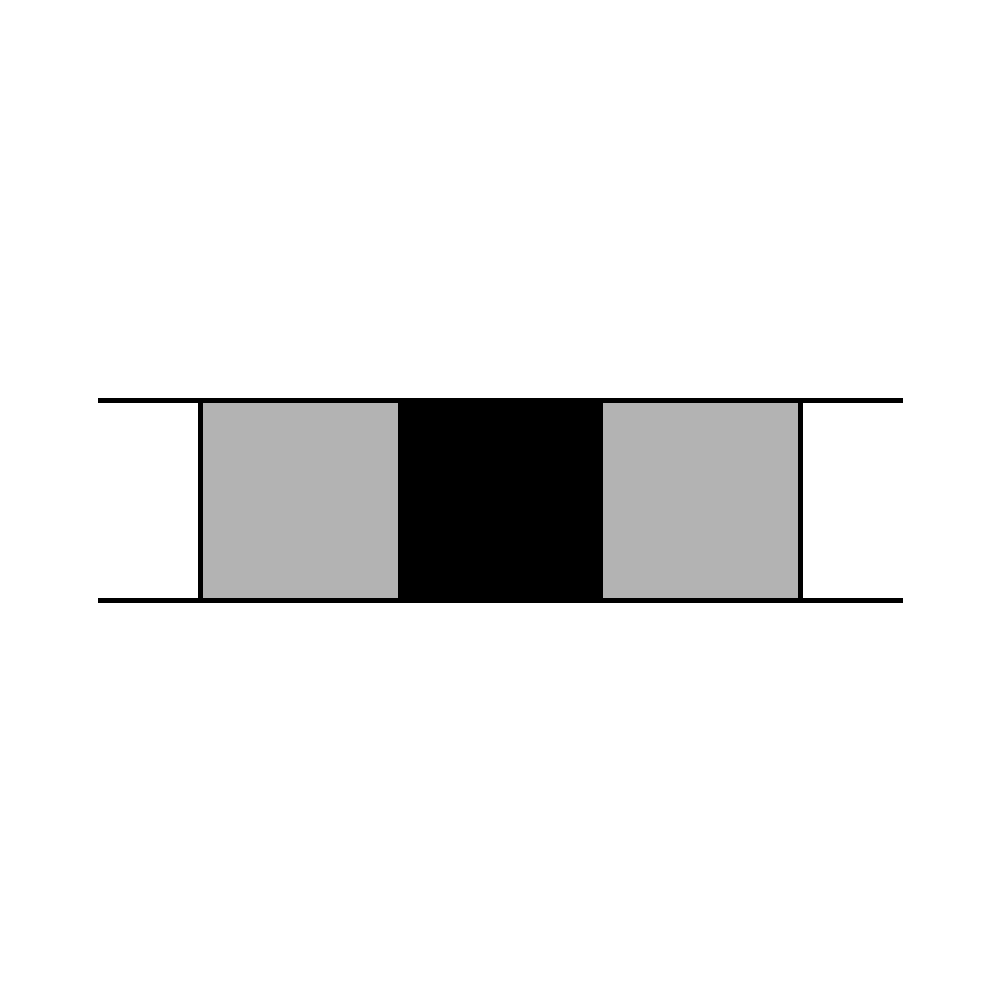}
    \includegraphics[width=0.4\linewidth]{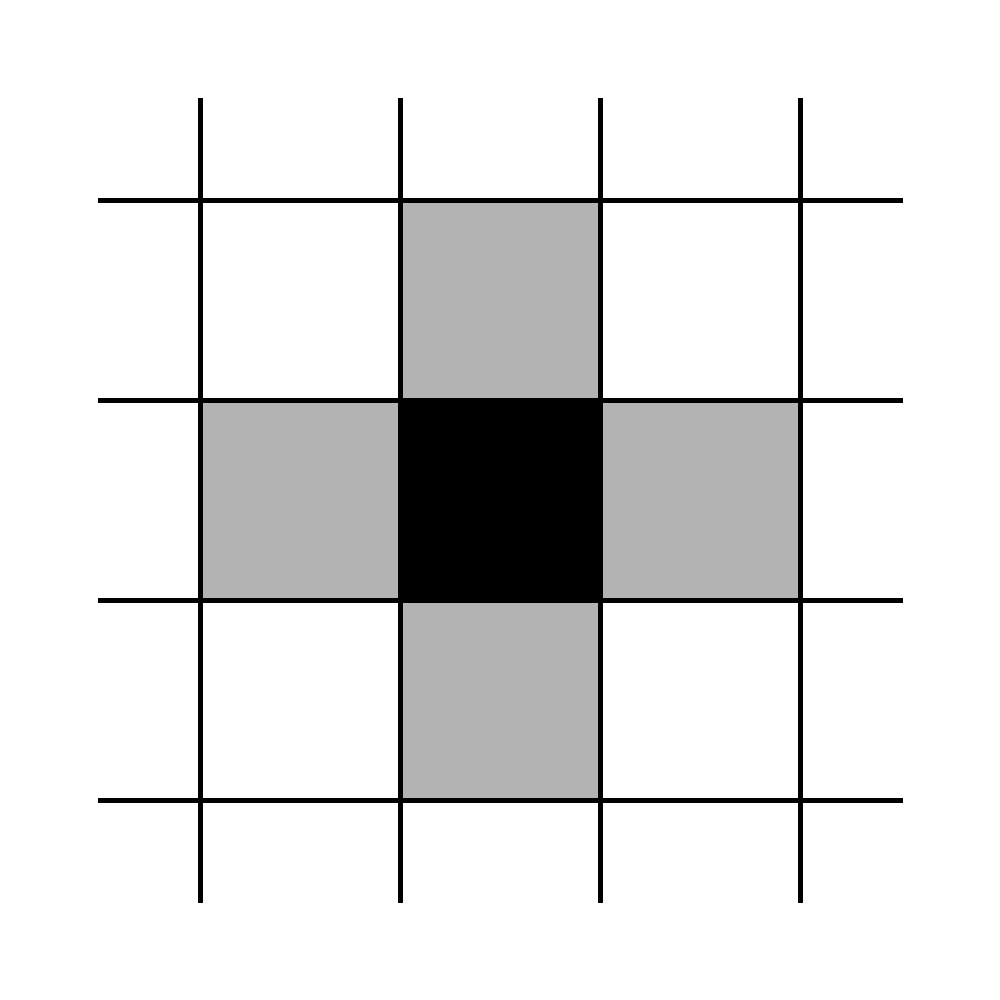}
    \caption{Neighborhoods (marked gray) of a site (black) in 1D and 2D.}
    \label{fig:nbhds}
\end{figure}

\subsection{Local Transition Probabilities}

We may summarize the rules above by specifying the conditional transition probabilities $w(s_i(t+1)|s_{i-1}(t),s_i(t),s_{i+1}(t))$ that the site $i$ changes from state $s_i(t)$ to $s_i(t+1)$ given that its neighbors are in states $s_{i-1}(t),s_{i+1}(t)$.
Since our states are only allowed to be $0$ or $1$, we must have $w(0|\mathbf{a}) = 1 - w(1|\mathbf{a})$ for any $\mathbf{a} \in \{0,1\}^3$, so it suffices to specify $w(1|\mathbf{a})$ for all $\mathbf{a} \in \{0,1\}^3$.
For the model described above, we find
\begin{equation}\label{eq:model}
    \begin{split}
        &w(1|000) = p_0, \quad w(1|001) = w(1|100) = p_1, \quad w(1|101) = p_2, \\
        &w(1|x1y) = 1 \text{ for all } x,y \in \{0,1\}.
    \end{split}
\end{equation}

We are interested in solving for the expected healthy site density of the monolayer for all times $t > 0$.
Equivalently, we would like to compute $P_t(0)$, the probability that an arbitrarily chosen site at time $t$ is 0 (healthy).
To do this, we must apply the local transition probabilities given by Eq. (\ref{eq:model}) to arbitrary probability measures on the space of states (for a general theory, see \cite{fuksConstructionLocalStructure2013}; for a worked example, see \cite{fuksExampleComputationDensity2018}).
In particular, the local transition probabilities given by Eq. (\ref{eq:model}) generalize the symmetric case considered by Fuk{\'s} \cite{fuksExampleComputationDensity2018} by introducing a non-zero ``spontaneous nucleation probability'' $p_0$, which substantially complicates the resulting analysis.
The proof of Theorem \ref{thm:main} is in turn inspired by the solution of the one-dimensional diffusion PCA by Fuk{\'s} \cite{fuksExampleComputationDensity2018}.

Define a \emph{block} $\mathbf{a}$ of length $n$ to be an element of $\{0,1\}^n$; we identify blocks with states of finite connected components of sites.
Furthermore, let $|\mathbf{a}|$ be the length of the block $\mathbf{a}$.
Since our model is translation-invariant, we need not specify where a block is placed, since all translates of a block have equal probability of occurrence.

Let $P_t(\mathbf{a})$ be the probability that an arbitrary block of length $|\mathbf{a}|$ at time $t$ is equal to $\mathbf{a}$.
Then, since the transition probabilities only depend on the neighborhood of a site, we may write the probability of a block of length $n$ at time $t+1$ in terms of all of the blocks of length $n+2$ at time $t$; in particular, we have
\[
    P_{t+1}(\mathbf{a}) = \sum_{\mathbf{b} \in \{0, 1\}^{|\mathbf{a}|+2}} w(\mathbf{a}|\mathbf{b}) P_t(\mathbf{b}),
\]
where
\[
    w(\mathbf{a}|\mathbf{b}) := \prod_{i=1}^{|\mathbf{a}|} w(a_i|b_ib_{i+1}b_{i+2}),
\]
is the probability that a fixed block transitions to $\mathbf{a}$ given that its current neighborhood is $\mathbf{b}$.

For example, using the transition probabilities given by Equation (\ref{eq:model}), we may compute
\begin{equation}\label{eq:n1unsimp}
    P_{t+1}(0)
        = (1 - p_0)P_t(0^3) + (1 - p_1)P_t(0^21) + (1 - p_1)P_t(10^2) + (1 - p_2)P_t(101),
\end{equation}
where superscripts indicate repetition; that is, $0^21$ indicates the block $001$.
However, not all the probabilities on the right hand side are independent.
In fact, for all $\mathbf{a}$ we have the \emph{consistency conditions} (c.f. \cite[Proposition~2.3]{fuksConstructionLocalStructure2013}),
\[
    P_t(\mathbf{a})
        = P_t(\mathbf{a}0) + P_t(\mathbf{a}1)
        = P_t(0\mathbf{a}) + P_t(1\mathbf{a}).
\]
These imply that, for all $n \geq 0$,
\begin{equation}\label{eq:consistency}
    \begin{split}
        P_t(0^{n+1}1)
            &= P_t(10^{n+1})
             = P_t(0^{n+1}) - P_t(0^{n+2}), \\
        P_t(10^{n+1}1)
            &= P_t(10^{n+1}) - P_t(10^{n+2})
             = P_t(0^{n+1}) - 2P_t(0^{n+2}) + P_t(0^{n+3}).        
    \end{split}
\end{equation}
Therefore, we may write the right hand side of Eq. (\ref{eq:n1unsimp}) in terms of $P_t(0^{n+1})$ only, for $n=0,1,2$, so it is natural to study the recurrence relations of $P_{t+1}(0^{n+1})$ for $n > 0$ as well.
In fact, we find
\begin{equation}\label{eq:ng1unsimp}
    \begin{split}
        P_{t+1}(0^{n+1}) = (1-p_0)^{n+1}&P_t(0^{n+3}) + (1-p_1)(1-p_0)^n\left(P_t(0^{n+2}1) + P_t(10^{n+2})\right) \\
        &+ (1-p_1)^2(1-p_0)^{n-1}P_t(10^{n+1}1).
    \end{split}
\end{equation}

Applying the consistency conditions (\ref{eq:consistency}) to Eqs. (\ref{eq:n1unsimp}) \& (\ref{eq:ng1unsimp}) leads us to the following:
\begin{proposition}\label{prop:recurrence}
    Let $P_t(0^{n+1})$ denote the probability that an arbitrarily chosen cluster of $n+1$ sites are all healthy at time $t$ in the probabilistic cellular automaton with transition probabilities given by Eq. (\ref{eq:model}).
    Then we may compute $P_{t+1}(0^{n+1})$ via the recurrence relations
    \begin{equation}\label{eq:recurrence}
        \begin{split}
            P_{t+1}(0)
            &= (1 - p_2)P_t(0) + 2(p_2 - p_1)P_t(0^2) - (p_0 - 2p_1 + p_2)P_t(0^3), \\
            P_{t+1}(0^{n+1})
            &= (1 - p_0)^{n-1}\left((1 - p_1)^2P_t(0^{n+1})\right. \\
            &\left.\qquad+ 2(1 - p_1)(p_1 - p_0)P_t(0^{n+2}) + (p_1 - p_0)^2P_t(0^{n+3})\right),
        \end{split}
    \end{equation}
    for $n > 0$.
\end{proposition}

\section{Explicit Solutions}

If we take the initial condition that $P_0(0^{n+1}) = 1$ for all $n \geq 0$, equivalent to the assumption that our monolayer is fully intact at time $t=0$, then it is clear that we can iterate Eq. (\ref{eq:recurrence}) to express $P_t(0^{n+1})$ as a polynomial in $p_0, p_1, p_2$.
The following theorem gives these polynomials explicitly in an important special case.

\begin{theorem}\label{thm:main}
    Suppose that $P_t(0^{n+1})$ obeys the recurrence relation given in Eq. (\ref{eq:recurrence}) with initial condition $P_0(0^{n+1}) = 1$ for all $n \geq 0$.
    Then
    \begin{equation}\label{eq:soln}
        P_t(0^{n+1}) = (1-p_0)^{(n-1)t} \prod_{k=0}^{t-1} \left((1-p_1) + (p_1-p_0)(1-p_0)^k\right)^2,
    \end{equation}
    for all $t \geq 0$ and $n > 0$.
    Furthermore, if $(1-p_0)(1-p_2)=(1-p_1)^2$, then Eq. (\ref{eq:soln}) also holds for $n = 0$.
\end{theorem}

In order to prove this theorem, we will need the following lemma from complex analysis.

\begin{lemma}\label{lemma:residue_sum}
    Let $P(z)$ be a polynomial of degree $n$, and $Q(z)$ be a polynomial of degree $m$.
    If $m \geq n + 2$, then the finite residues of the rational function
    \[
        f(z) := \frac{P(z)}{Q(z)},
    \]
    sum to zero.
    That is,
    \[
        \sum_{a_k}\Res_{z = a_k} f = 0,
    \]
    where the sum is taken over the poles $a_k$ of $f$ (the zeros of $Q$).
    
    Equivalently, we have that
    \[
        \Res_{z = \infty} f = 0.
    \]
\end{lemma}

\begin{proof}
    Define
    \[
        I_R := \left|\frac{1}{2\pi i} \int_{|z| = R} f(z) \xdif z\right|
    \]
    Since $m \geq n + 2$, there exists $C \in \RR$ such that for $|z|$ sufficiently large we have
    \[                                                                                                                                                                                              
        |f(z)| \leq \frac{C}{|z|^2}.
    \]
    Therefore, for $R$ sufficiently large, we may bound $I_R$ by
    \begin{equation}\begin{split}\label{eq:arc_estimate}
        I_R
            &= \left|\frac{1}{2\pi i} \int_{|z| = R} f(z) \xdif z\right| \\
            &\leq \frac{1}{2\pi} \int_{|z| = R} |f(z)| |\xdif z| \\
            &\leq \frac{1}{2\pi} \int_{|z| = R} \frac{C}{|z|^2} |\xdif z| \\
            &= \frac{1}{2\pi} \frac{C}{R^2} (2\pi R) \\
            &= \frac{C}{R}.
    \end{split}\end{equation}
    On the other hand, for $R$ sufficiently large, the contour $\{|z| = R\}$ will enclose all of the poles of $f$, so we may compute $I_R$ via the residue theorem (see, e.g., \cite{ahlforsComplexAnalysisIntroduction1979}).
    In particular, we have
    \[
        I_R = \sum_{a_k} \Res_{z = a_k} f(z).
    \]
    Pairing this with Eq. (\ref{eq:arc_estimate}), and sending $R \to \infty$, we conclude the proof.
\end{proof}

\begin{proof}[Proof of Theorem \ref{thm:main}]
    Suppose first that $(1-p_0)(1-p_2)=(1-p_1)^2$, so that the two recurrences of Eq. (\ref{eq:recurrence}) simplify to
    \begin{equation}\label{eq:coeff_simple}
        \begin{split}
            P_{t+1}(0^{n+1})
                &= (1 - p_0)^{n-1}\left((1 - p_1)^2P_t(0^{n+1})\right. \\
                &\left.\qquad+ 2(1 - p_1)(p_1 - p_0)P_t(0^{n+2}) + (p_1 - p_0)^2P_t(0^{n+3})\right),
        \end{split}
    \end{equation}
    for all $n \geq 0$.
    
    \proofpart{Generating Functions}

    For $t \geq 0$, define the formal power series $\mathbf{P}_t \in \RR[[z]]$ by $\mathbf{P}_t(z) := \sum_{n \geq 0} P_t(0^{n+1})z^n$.
    Conversely, we have $[z^n] \mathbf{P}_t(z) = P_t(0^{n+1})$, where the operator $[z^n]F(z)$ extracts the coefficient of $z^n$ of the power series $F$.
    We would like to transform Equation (\ref{eq:coeff_simple}) from a family of recurrences in $P_t(0^{n+1})$ to a recurrence of $\mathbf{P}_t$.
    
    For $n \geq 0$, we may compute
    \[\begin{split}
        [z^n] \mathbf{P}_{t+1}(z)
            &= P_{t+1}(0^{n+1}) \\
            &= (1 - p_0)^{n-1}\left((1 - p_1)^2P_t(0^{n+1})\right. \\
            &\left.\qquad+ 2(1 - p_1)(p_1 - p_0)P_t(0^{n+2}) + (p_1 - p_0)^2P_t(0^{n+3})\right) \\
            &= (1 - p_0)^{n-1}\left((1 - p_1)^2 [z^n]\mathbf{P}_t(z)\right. \\
            &\left.\qquad+ 2(1 - p_1)(p_1 - p_0)[z^{n+1}]\mathbf{P}_t(z) + (p_1 - p_0)^2[z^{n+2}]\mathbf{P}_t(z)\right).
    \end{split}\]
    Using the coefficient extraction property $[z^{n+k}]F(z) = [z^n]z^{-k}F(z)$, we may further simplify this to
    \[\begin{split}
        [z^n] \mathbf{P}_{t+1}(z)
            &= (1 - p_0)^{n-1}[z^n]\left((1 - p_1)^2 + 2(1 - p_1)(p_1 - p_0)z^{-1} + (p_1 - p_0)^2z^{-2}\right)\mathbf{P}_t(z) \\
            &= (1 - p_0)^n[z^n]\frac{(1 - p_1)^2 + 2(1 - p_1)(p_1 - p_0)z^{-1} + (p_1 - p_0)^2z^{-2}}{1 - p_0}\mathbf{P}_t(z) \\
            &= (1 - p_0)^n[z^n]\frac{\left((1 - p_1) + (p_1 - p_0)z^{-1}\right)^2}{1 - p_0}\mathbf{P}_t(z).
    \end{split}\]
    
    Let $g(z) := \frac{\left((1 - p_1) + (p_1 - p_0)z^{-1}\right)^2}{1 - p_0} = \frac{\left((1 - p_1)z + (p_1 - p_0)\right)^2}{(1 - p_0)z^2}$ and $\beta := 1 - p_0$, so that we may rewrite the above as
    \begin{equation}\label{eq:coeffs}
        [z^n] \mathbf{P}_{t+1}(z)
            = \beta^n [z^n] g(z) \mathbf{P}_t(z)
            = [z^n] g(\beta z) \mathbf{P}_t(\beta z).
    \end{equation}

    Equation (\ref{eq:coeffs}) is valid for $n \geq 0$.
    We would like to extend this to hold for all $n \in \ZZ$.
    Since $g$ contains no terms of positive degree, if $\mathbf{P}_t$ were to have terms of negative degree then they would not affect the terms of non-negative degree in the product on the right-hand side of Equation (\ref{eq:coeffs}).
    Therefore, we may consider Eq. (\ref{eq:coeffs}) an equality of power series so long as we consider only terms of non-negative degree in the end.
    
    In fact, formally one may view $\mathbf{P}_t$ as an element of the Hardy space on the circle, and the right-hand side of Equation (\ref{eq:coeffs}) as a composition of Toeplitz and dilation operators (see, e.g., \cite{axlerToeplitzOperators2010}).
    With this view, we may rewrite Eq. (\ref{eq:coeffs}) as
    \[
        \mathbf{P}_{t+1} = D_\beta T_g \mathbf{P}_t,
    \]
    and the considerations of the previous paragraph are equivalent to the fact that $T_fT_g = T_{fg}$ for functions $f,g$ with no terms of positive degree (c.f. \cite[Theorem~2]{axlerToeplitzOperators2010}).
    
    Iterating Eq. (\ref{eq:coeffs}), for $n,t \geq 0$ we get
    \[
        [z^n] \mathbf{P}_t(z) = [z^n] \mathbf{P}_0(\beta^t z) \prod_{m=1}^t g(\beta^m z).
    \]
    Our initial condition $P_0(0^{n+1}) = 1$ gives that $\mathbf{P}_0(z) = \frac{1}{1 - z}$, so we conclude
    \begin{equation}\label{eq:coeff_fn}
        P_t(0^{n+1}) = [z^n] \frac{1}{1 - \beta^t z} \prod_{m=1}^t g(\beta^m z),
    \end{equation}
    for $n,t \geq 0$.

    \proofpart{Coefficient Extraction}

    By, e.g., Cauchy's Theorem \cite{ahlforsComplexAnalysisIntroduction1979} we may recognize the $z^{-1}$ coefficient of a power series $F$ as the residue of $F$ at zero; that is, we may write
    \[
        [z^n]F(z) = \Res_{z=0} \frac{F(z)}{z^{n+1}}.
    \]
    Therefore, Eq. (\ref{eq:coeff_fn}) may be rewritten as
    \begin{equation}\label{eq:coeff_res}
        P_t(0^{n+1}) = \Res_{z=0} \frac{1}{z^{n+1}(1 - \beta^t z)} \prod_{m=1}^t g(\beta^m z),
    \end{equation}
    so we have reduced the problem to that of computing the residue at zero of the rational function
    \[
        f(z) := \frac{1}{z^{n+1}(1 - \beta^t z)} \prod_{m=1}^t g(\beta^m z).
    \]

    Since the order of the pole of $f$ at zero is $2t + n + 1$, computing the residue at zero is prohibitively difficult.
    On the other hand, $f$ has a single other pole at $z = \beta^{-t}$ that is a simple pole.
    Since we may write $f(z) = \frac{P(z)}{Q(z)}$ for polynomials $P,Q$ with order $2t, 2t + n + 2$, respectively, we may use Lemma \ref{lemma:residue_sum} to relate these two residues.
    In particular, we must have
    \begin{equation}\label{eq:residue}
        \Res_{z=0} f(z) = -\Res_{z=\beta^{-t}} f(z),
    \end{equation}
    which is a much simpler computation; in fact, by the definition of residues for simple poles, it is the limit
    \[
        \Res_{z=\beta^{-t}} \frac{1}{z^{n+1}(1-\beta^t z)}\prod_{m=1}^t g(\beta^m z)
            = \lim_{z \to \beta^{-t}} \frac{z - \beta^{-t}}{z^{n+1}(1-\beta^t z)}\prod_{m=1}^t g(\beta^m z).
    \]
    Simplification and direct substitution gives
    \[\begin{split}
        \lim_{z \to \beta^{-t}} \frac{z - \beta^{-t}}{z^{n+1}(1-\beta^t z)}\prod_{m=1}^t g(\beta^m z)
            &= \lim_{z \to \beta^{-t}} \frac{\beta^t z - 1}{\beta^t z^{n+1}(1-\beta^t z)}\prod_{m=1}^t g(\beta^m z) \\
            &= -\beta^{nt} \prod_{m=1}^t g(\beta^{m-t})
             = -\beta^{nt} \prod_{m=0}^{t-1} g(\beta^{-m}).
    \end{split}\]
    Substituting this back into Eq. (\ref{eq:coeff_res}) by way of Eq. (\ref{eq:residue}), we conclude that
    \[\begin{split}
        P_t(0^{n+1})
            &= \beta^{nt} \prod_{m=0}^{t-1} g(\beta^{-m}) \\
            &= (1 - p_0)^{nt} \prod_{m=0}^{t-1} \frac{((1 - p_1) + (p_1 - p_0)(1 - p_0)^m)^2}{1 - p_0} \\
            &= (1 - p_0)^{(n-1)t} \prod_{m=0}^{t-1} ((1 - p_1) + (p_1 - p_0)(1 - p_0)^m)^2.
    \end{split}\]

    To conclude the proof, it remains to show that this expression still holds for $n > 0$ in case $(1-p_0)(1-p_2) \neq (1-p_1)^2$, but this is a simple consequence of the fact that this change only affects the recursion for the term $P_{t+1}(0)$.
    Since $P_{t+1}(0^{n+1})$ only depends on $P_t(0^{n+1}), P_t(0^{n+2})$, and $P_t(0^{n+3})$, making this change does not affect any other terms for $n > 0$.
\end{proof}

We may interpret the additional hypothesis of Theorem \ref{thm:main} as the requirement that the probabilities $1-p_i$ follow a geometric distribution; that is,
\[
    p_i = 1 - p_0 \exp(\alpha i),
\]
for some $\alpha \in \RR$.
Instead, the following theorem may be used for general $p_i$.

\begin{theorem}\label{thm:general}
    Let $A_t := P_t(0)$ be the expected density of healthy monolayer after $t$ iterations of the PCA rule defined in Eq. (\ref{eq:model}).
    Then
    \[
        A_t = (1-p_2)^t\left(1+\sum_{m=0}^{t-1} (1-p_2)^{-m-1}\left(2(p_2-p_1) -(p_0-2p_1+p_2)(1-p_0)^m\right) C_m\right),
    \]
    where
    \[
        C_m = \prod_{k=0}^{m-1} \left[(1-p_1) + (p_1-p_0)(1-p_0)^k\right]^2.
    \]
\end{theorem}

\begin{proof}
    By Eq. (\ref{eq:recurrence}), $A_t = P_t(0)$ obeys the recurrence
    \begin{equation}\label{eq:general1}
        A_{t+1} = (1 - p_2)A_t + 2(p_2 - p_1)P_t(0^2) - (p_0 - 2p_1 + p_2)P_t(0^3),
    \end{equation}
    where $P_t(0^2), P_t(0^3)$ are now explicit functions of $t$ by Theorem \ref{thm:main}.
    Such recurrences are standard in the literature (see, e.g., \cite{elaydiIntroductionDifferenceEquations2005}), but for completeness we give a full proof.

    To eliminate the coefficient $(1 - p_2)$, we first multiply both sides of Eq. (\ref{eq:general1}) by $(1 - p_2)^{-t-1}$ and set $\tilde{A}_t := (1 - p_2)^{-t} A_t$ to get
    \[
        \tilde{A}_{t+1} = \tilde{A}_t + (1 - p_2)^{-t-1}\left(2(p_2 - p_1)P_t(0^2) - (p_0 - 2p_1 + p_2)P_t(0^3)\right).
    \]
    This may be solved easily using a telescoping sum and Theorem \ref{thm:main} as
    \[\begin{split}
        \tilde{A}_t
            &= \tilde{A}_0 + \sum_{m=0}^{t-1} \left[\tilde{A}_{m+1} - \tilde{A}_m\right] \\
            &= 1 + \sum_{m=0}^{t-1} (1 - p_2)^{-m-1}\left(2(p_2 - p_1)P_m(0^2) - (p_0 - 2p_1 + p_2)P_m(0^3)\right) \\
            &= 1+\sum_{m=0}^{t-1} (1-p_2)^{-m-1}\left(2(p_2-p_1) -(p_0-2p_1+p_2)(1-p_0)^m\right) C_m.
    \end{split}\]
    Substituting back in $A_t = (1 - p_2)^t \tilde{A}_t$, we conclude.
\end{proof}

\section{The Continuous-Time Model}\label{sec:cont}

Rather than a discrete model as considered in Section \ref{sec:model}, we may consider a continuous-time model.
In particular, we may replace Rule 2 in our model by the following continuous analog:
\begin{enumerate}
    \setcounter{enumi}{1}
    \item Each cell degrades after an exponentially distributed time with rate parameter $k_n$, where $n$ is the number of degraded neighbors of the site.
\end{enumerate}

With this rule, we may then view the discrete dynamical system described in Section \ref{sec:model} as a discretization of the continuous-time dynamical system here as follows: according to Rule 2, for $\Delta t$ sufficiently small, the probability of a degradation event in the window $[t, t + \Delta t]$ for a site with $i$ degraded neighbors is given by
\[
    p_i = 1 - e^{-k_i \Delta t} + o(\Delta t),
\]
for rate parameters $k_i \geq 0, i = 0, 1, 2$.
The remainder term comes from the fact that there is a small probability that in the window $[t, t + \Delta t]$, a neighbor will degrade and change the rate parameter from $k_i$ to $k_{i+1}$.
However, since we are interested in the limit as $\Delta t \to 0$, we may neglect this term for the analysis of the limit expression.
We may then approximate the degraded density at time $t$ by $A_{\floor{t/\Delta t}} = P_{\floor{t/\Delta t}}(0)$ and recover the continuous-time dynamics by sending $\Delta t \to 0$.

\begin{theorem}\label{thm:cont_simple}
    Let $p_i = 1 - e^{-k_i \Delta t}$ for $k_i \geq 0$, and suppose that $2k_1 = k_0 + k_2$.
    Then as $\Delta t \to 0$, $A(t) := A_{\floor{t/\Delta t}}$ converges to the function
    \begin{equation}\label{eq:poisson}
        A(t) = e^{(k_0 - 2k_1)t + \frac{2(k_1 - k_0)}{k_0}(1 - e^{-k_0 t})}.
    \end{equation}
\end{theorem}

\begin{proof}
    Note first that the condition $2k_1 = k_0 + k_2$ is equivalent to $(1-p_1)^2 = (1-p_0)(1-p_2)$, so we may use Theorem \ref{thm:main} to solve for $A_{\floor{t/\Delta t}}$.
    At time $t = n \Delta t$, Theorem \ref{thm:main} with the choice of Poisson-distributed $p_i$ gives
    \[\begin{split}
        \log A(t)
            &= \log A_n \\
            &= k_0 n \Delta t + \log \prod_{j=0}^{n-1}\left(e^{-k_1 \Delta t} + (e^{-k_0 \Delta t} - e^{-k_1 \Delta t}e^{-k_0 j \Delta t})\right)^2 \\
            &= k_0 t + 2\sum_{j=0}^{n-1} \log\left(e^{-k_1 \Delta t} + (e^{-k_0 \Delta t} - e^{-k_1 \Delta t})e^{-k_0 j \Delta t}\right) \\
            &= k_0 t + 2\sum_{j=0}^{n-1} \log\left(1 - (k_1 + (k_0 - k_1)e^{-k_0 j \Delta t}) \Delta t + O((\Delta t)^2)\right) \\
            &= k_0 t - 2\sum_{j=0}^{n-1} \left((k_1 + (k_0 - k_1)e^{-k_0 j \Delta t}) \Delta t + O((\Delta t)^2)\right).
    \end{split}\]
    As $\Delta t \to 0$, the Riemann sum becomes an integral and the higher order terms vanish, giving
    \[\begin{split}
        \log A(t)
            &= k_0 t - 2 \int_0^t (k_1 + (k_0 - k_1)e^{-k_0 s}) \xdif s \\
            &= (k_0 - 2k_1)t + \frac{2(k_0 - k_1)}{k_0}(e^{-k_0 t} - 1) \\
            &= (k_0 - 2k_1)t + \frac{2(k_1 - k_0)}{k_0}(1 - e^{-k_0 t}).
    \end{split}\]
    Taking the exponential of both sides, we get Eq. (\ref{eq:poisson}).
\end{proof}

\begin{figure}
    \centering
    \includegraphics[width=0.95\linewidth]{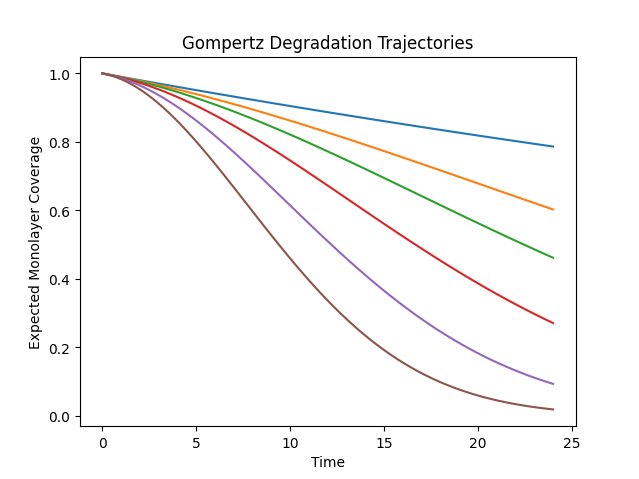}
    \caption{Degradation trajectories computed via Equation (\ref{eq:poisson}) with $k_i = 0.01 + \alpha i$ for values of $\alpha \in [0, 0.7]$.}
    \label{fig:gompertz}
\end{figure}

The function $A(t)$ takes the form of the product of an exponential decay and a Gompertz-type decay \cite{tjorveUseGompertzModels2017}, giving $A(t)$ a characteristic sigmoid shape (see Figure \ref{fig:gompertz}).
In particular, we may compute the \emph{hazard rate} $h(t)$ associated with the decay process $A(t)$:
\begin{equation}\label{eq:hazard}
    h(t) := -\frac{A'(t)}{A(t)} = -\frac{\xdif}{\xdif t} \log A(t).
\end{equation}
The hazard rate $h(t)$ gives the instantaneous decay rate for the process and may be thought of as an \emph{effective} $k$ for the system as a whole.
Substituting Eq. (\ref{eq:poisson}) into Eq. (\ref{eq:hazard}), we get the hazard rate
\begin{equation}\label{eq:hazard_simple}
    h(t) = k_0 + 2(k_1 - k_0)(1 - e^{-k_0 t}).
\end{equation}

Therefore, at time $t=0$, the system degrades with rate $k_0$, which is natural since at time $t=0$ the monolayer is fully intact and there are no defects.
As $t$ increases, $h(t)$ interpolates between $k_0$ and $2k_1 - k_0 = k_2$ at $t=\infty$.
Again, this is to be expected since at $t=\infty$, the monolayer is fully degraded.
Therefore, Eq. (\ref{eq:hazard_simple}) gives a quantitative control on the rate at which the system degrades as a whole.

As an example of Theorem \ref{thm:cont_simple}, if $k_i \equiv k \geq 0$, then the local interactions disappear and Theorem \ref{thm:cont_simple} reduces to
\[
    A(t) = e^{-kt},
\]
which is the coverage equation one gets with the Polanyi-Wigner equation \cite{ligterinkOverviewDesorptionParameters2023} of order $1$:
\[
    \frac{\xdif A}{\xdif t} = -kA,
\]
where $k=ve^{\frac{-E_a}{RT}}$ is distributed according to the Arrhenius equation.
In practice, one would expect the rate parameters $k_i$ to include a dependence on $i$ in the activation energy $E_a$:
\begin{equation}\label{eq:arrhenius}
    k_i = ve^{\frac{-E_a(i)}{RT}},
\end{equation}
where $v$ is the pre-exponential factor, $R$ is the molar gas constant, $T$ is temperature, and $E_a(i)$ is the neighbor-dependent activation energy of the desorption reaction.
If each degraded neighbor removes a fixed amount of energy to the required activation energy of the reaction, then we can express $E_a(i)$ as
\begin{equation}\label{eq:activation_energy}
    E_a(i) = E_0 - i \delta E,
\end{equation}
for some (tunable) parameters $E_0, \delta E$ depending on the experimental parameters.

Substituting Eq. (\ref{eq:activation_energy}) into Eq. (\ref{eq:arrhenius}), we get the rate parameters
\[
    k_i = ve^{\frac{-E_0}{RT}}e^{\frac{i\delta E}{RT}}.
\]
In this case, the $k_i$ are in geometric sequence, so we cannot, in general, apply Theorem \ref{thm:cont_simple}.
We need the following theorem, the continuous analog of Theorem \ref{thm:general}.

\begin{theorem}\label{thm:cont_general}
    Let $p_i = 1 - e^{-k_i \Delta t}$ for $k_i \geq 0$.
    Then as $\Delta t \to 0$, $A(t) := A_{\floor{t/\Delta t}}$ converges to the function
    \begin{equation}\label{eq:poisson_general}
        A(t) = e^{-k_2 t}\left(1 + \int_0^t (2(k_2 - k_1) - (k_0 - 2k_1 + k_2)e^{-k_0 s})e^{(k_2 - 2k_1)s + \frac{2(k_1 - k_0)}{k_0}(1 - e^{-k_0 s})} \xdif s\right)\hspace{-1pt}.
    \end{equation}
\end{theorem}

\begin{proof}
    First, define the quantity
    \[
        \tilde{A}_n
            := e^{k_0 n\Delta t}\prod_{m=0}^{n-1}\left(e^{-k_1 \Delta t} + (e^{-k_0 \Delta t} - e^{-k_1 \Delta t})e^{-k_0 m \Delta t}\right)^2.
    \]
    By Theorem \ref{thm:cont_simple}, we have that
    \begin{equation}\label{eq:Atildelimit}
        \tilde{A}(t) := \tilde{A}_{\floor{t/\Delta t}} \to e^{(k_0-2k_1)t + \frac{2(k_1-k_0)}{k_0}(1 - e^{-k_0 t})},
    \end{equation}
    as $\Delta t \to 0$.
    Substituting the probabilities $p_i = 1 - e^{-k_i \Delta t}$ into the expression for $A_n$ given by Theorem \ref{thm:general} and using $\tilde{A}_n$ to simplify the expression, at time $t = n\Delta t$, we find
    \[\begin{split}
        A(n \Delta t)
            &= A_n \\
            &= e^{-k_2 t}\left(1 + \sum_{m=0}^{n-1} e^{(k_2 - k_0) m \Delta t}\left(2(k_2 - k_1) - (k_0 - 2k_1 + k_2)e^{-k_0 m \Delta t}\right)\tilde{A}(m \Delta t) \Delta t\right. \\
            &\left.\phantom{\sum_{m=0}^{n-1}} + O((\Delta t)^2) \right).
    \end{split}\]
    As in the proof of Theorem \ref{thm:cont_simple}, we recognize the sum as a Riemann sum.
    Therefore, using Eq. (\ref{eq:Atildelimit}), as $\Delta t \to 0$ we recover Eq. (\ref{eq:poisson_general}).
\end{proof}

\section{Discussion}

Self-assembled monolayers are a fundamental tool in the construction of modern biosensors, and understanding properties of their degradation is crucial in determining the lifetime of these biosensors.
In this work we have provided a probabilistic cellular automaton model for the degradation of self-assembled monolayers that generalizes the Polanyi-Wigner equation by including adsorbate-adsorbate interactions.
Furthermore, we solved for the expected degraded area explicitly in one dimension using methods from generating function theory and complex analysis.

The use of complex analysis in the study of recurrence equations and generating functions is not new.
Analytic combinatorics, as pioneered by Flajolet \& Sedgewick \cite{flajoletAnalyticCombinatorics2009}, uses complex analytic information about the generating function to extract asymptotic information about the growth rate of the coefficients.
However, by utilizing exact knowledge of the location of all of the poles, the proof of Theorem \ref{thm:main} computes the coefficients exactly by leveraging the rigidity of meromorphic functions.

In fact, by \emph{mutatis mutandi} repeating the proof of Theorem \ref{thm:main}, one can prove the following theorem on (possibly infinite) linear recurrences:

\begin{theorem}
    Suppose $g_n(z) := \sum_{k=0}^\infty a_k^{(n)}z^k$ are entire functions and $b \in \RR$.
    If $f_m^{(n)}$ is a two-parameter sequence satisfying
    \begin{align*}
        f_m^{(n+1)}
            &= b^m \sum_{k=0}^\infty a_k^{(n)} f_{m+k}^{(n)}, \\
        f_m^{(0)}
            &= 1,
    \end{align*}
    then we must have
    \[
        f_m^{(n)} = b^{mn}\prod_{k=1}^n g_{k-1}(b^{n-k}).
    \]
\end{theorem}

As mentioned in the introduction, the method presented in this work does not immediately generalize to the two dimensional case, since the generating function method requires writing out a recurrence relation over the clusters of a given size on the lattice.
Since such an enumeration of clusters is unavailable in two dimensions, this recursion cannot be written out in full generality.
In theory, we could use the method of Redelmeier \cite{redelmeierCountingPolyominoesAnother1981} to enumerate clusters up to size $n=24$, but this would not inform us about the medium- to long-term behavior of the monolayer which is the most important for applications.
However, we may instead view the one-dimensional result as an interpolation between the Polanyi-Wigner equation and the true two-dimensional behavior.
In future work, it will be interesting to try to fit experimental data to Eq. (\ref{eq:poisson_general}) and compare the values for $k_i$ with those predicted by the intermolecular kinetics.

\section*{Disclaimer}

This paper is an official contribution of the National Institute of Standards and Technology and is not subject to copyright in the United States.
Commercial products are identified in order to adequately specify certain procedures.
In no case does such identification imply recommendation or endorsement by the National Institute of Standards and Technology, nor does it imply that the identified products are necessarily the best available for the purpose.

\section*{Data Availability Statement}

We do not analyse or generate any datasets, because our work proceeds within a theoretical and mathematical approach.

\section*{Conflicts of Interest}

All authors declare no conflicts of interest in this paper.

\section*{Acknowledgments}

This work was performed with generous support from the CHIPS and Science Act NIST Grant 2.12.

\bibliography{ProbabilisticCellularAutomata}

\end{document}